\documentclass[11pt]{amsart}
\usepackage{amssymb,amsxtra,amsmath,amsfonts}
\allowdisplaybreaks
\newcommand{\p}[1]{\partial_{#1}}
\newcommand{\Z}{\mathbb{Z}}

\newcommand{\M}{\bar{M}}
\newcommand{\W}{\bar{W}}
\newcommand{\inv}{^{-1}}
\newtheorem{theorem}{Theorem}
\newtheorem{lemma}[theorem]{Lemma}
\newtheorem{proposition}[theorem]{Proposition}

\theoremstyle{definition}
\newtheorem{definition}[theorem]{Definition}
\newtheorem{example}[theorem]{Example}

\theoremstyle{remark}
\newtheorem{remark}[theorem]{Remark}

\numberwithin{equation}{section}
\numberwithin{theorem}{section}

\newcommand\leref{Lemma \ref}
\newcommand\prref{Proposition \ref}

\newcommand\exref{Example \ref}
\newcommand\reref{Remark \ref}
\newcommand\seref{Section \ref}

\DeclareMathOperator\Der{Der}

\DeclareMathOperator\Ker{Ker}

\DeclareMathOperator\Span{span}

\def\A{\mathcal{A}}
\def\DCP{\mathcal{D}_{\mathbb{CP}^1}}
\def\F{\mathcal{F}}
\def\MM{\mathcal{M}}

\def\E{\mathcal{E}}
\def\T{\mathcal{T}}
\def\V{\mathcal{V}}
\def\VC{\mathcal{V}^\circ}
\def\ps{\partial_s}
\def\dd{\partial^*}
\def\dh{\hat{\partial}}
\def\dc{\partial^\circ}
\def\fin{\mathrm{fin}}
\def\x{\boldsymbol{x}}
\def\t{\boldsymbol{t}}
\def\bt{\bar{\t}}
\def\CC{\mathbb{C}}
\def\al{\alpha}

\def\Ga{\Gamma}
\def\de{\delta}

\def\La{\Lambda}

\def\d{\partial}

\begin{document}

\title[Additional symmetries of the extended bigraded Toda]{Additional symmetries of the extended bigraded Toda hierarchy}
\author{Bojko Bakalov}
\author{William Wheeless}
\address{Department of Mathematics,
North Carolina State University,
Raleigh, NC 27695, USA}
\email{bojko\_bakalov@ncsu.edu; wbwheele@ncsu.edu}


\date{July 27, 2015; Revised December 3, 2015}



\begin{abstract}
The extended bigraded Toda hierarchy (EBTH) is an integrable system satisfied by the total descendant potential of $\mathbb{CP}^1$ with two orbifold points. We construct additional symmetries of the EBTH and describe explicitly their action on the Lax operator, wave operators, and tau-function of the hierarchy. In particular, we obtain infinitesimal symmetries of the EBTH that act on the tau-function as a subalgebra of the Virasoro algebra, generalizing those of Dubrovin and Zhang.
\end{abstract}

\maketitle


\section{Introduction}\label{s1}

The \emph{extended Toda hierarchy} was introduced by E. Getzler \cite{Ge01} and Y. Zhang \cite{Z02} 
in bihamiltonian form and by Carlet--Dubrovin--Zhang \cite{CDZ04} in Lax form, as an integrable
hierarchy that governs the Gromov--Witten theory of $\mathbb{CP}^1$. 
Recall that the \emph{total descendant potential} $\mathcal{D}_X$ of a projective manifold (or more generally, orbifold) $X$
is a certain generating series of the Gromov--Witten invariants of $X$ (see \cite{Giv01}).
A version of the so-called Toda conjecture \cite{EY94,Ge01,Z02}
states that $\DCP$ is a tau-function of the extended Toda hierarchy
(and so satisfies Milanov's Hirota quadratic equations \cite{Mil07}).
There are now several known proofs of this conjecture; see \cite{Ge01,DZ04,Mil06,OP06}.
On the other hand, by the Virasoro conjecture \cite{EHX97} proved by A. Givental \cite{Giv01},
we have $L_n\DCP=0$ for $n\ge-1$, where the operators $L_n$ give a representation of (a subalgebra of)
the Virasoro algebra. It was shown by Dubrovin--Zhang \cite{DZ04} that the operators $L_n$ are infinitesimal symmetries
of the extended Toda hierarchy when acting on its tau-function. 

The \emph{extended bigraded Toda hierarchy} (EBTH) was introduced by G. Carlet \cite{Car06} as a generalization
of the extended Toda hierarchy related to the Frobenius manifolds from \cite{DZ98}. 
There is an EBTH for every pair $(k,m)$ of positive integers, with $k=m=1$ corresponding to the
extended Toda hierarchy. The total descendant potential of $\mathbb{CP}^1$
with two orbifold points of orders $k$ and $m$ is a tau-function of the EBTH (see \cite{MT08,CvdL13}).
Note that a part of the EBTH is the bigraded Toda hierarchy, which is a reduction of the 2D Toda hierarchy
(see \cite{Tak10,UT84}). 

Let us recall the additional symmetries of the KP hierarchy, following \cite{OS86,Di03,vM91}.
One introduces a pseudo-differential operator $M$ such that
$[L,M]=1$ for the Lax operator $L$; then for any $p\in\Z$, $n\in\Z_{\ge0}$, 
the operators $L^p M^n$ induce vector fields that commute with the flows of the KP hierarchy.
The action of these symmetries on the tau-function is given by the
\emph{Adler--Shiota--van Moerbeke formula} (see \cite{ASvM95,Di95,Di03,vM91}).
One obtains a representation of the $W_{1+\infty}$-algebra, which contains as subalgebras
the Heisenberg algebra (corresponding to $L^p$) and the Virasoro algebra (corresponding to $L^pM$).
Similar results hold also for the 2D Toda hierarchy \cite{ASvM95}. However, for their reductions such as
the KdV hierarchy or 1D Toda hierarchy, one has only the Virasoro algebra and not the whole
$W_{1+\infty}$-algebra (see \cite{AvM95,Di03,vM91}).

Now we summarize the contents of the present paper.
In \seref{s2}, we review the extended Toda hierarchy and 
the extended bigraded Toda hierarchy (EBTH), following the approach of K. Takasaki \cite{Tak10}.
Our version of the EBTH is related to the original definition of G. Carlet \cite{Car06} (or to \cite{CvdL13})
by an explicit change of variables, and we believe it is more convenient.
We discuss the Lax operator $L$, the wave operators, and wave functions of the EBTH.

Then, in \seref{s3}, we find two operators $M$ and $\M$ such that $[L,M]=[L,\M]=1$, 
and use them to construct additional symmetries of the EBTH. 
We describe explicitly the action of these symmetries on the Lax and wave operators.
We show that the symmetries corresponding to $L^p(M-\M)$ commute as the Virasoro 
vector fields $-z^p \p{z}$ ($p\ge0$). 
Our operators $M$ and $\M$ extend those of \cite{LHS12}, which were used there to 
construct additional symmetries of the bigraded Toda hierarchy. 
Unfortunately, the papers \cite{LHS12,LH13} have flaws; in particular, 
the higher powers of $M-\M$ used there are not well defined (see \reref{rmmbar} below).

In \seref{s4}, we determine how the additional symmetries of the EBTH 
act on the tau-function, thus getting an analog of the Adler--Shiota--van Moerbeke formula.
This result is new even in the case $k=m=1$ corresponding to the extended Toda hierarchy.
In this case, our Virasoro operators coincide with those of \cite{DZ04} (see \exref{evir2} below).
Dubrovin and Zhang \cite{DZ04}  have shown that these operators are infinitesimal symmetries
of the extended Toda hierarchy when acting on the tau-function, but their action
on the Lax and wave operators was previously unknown. 
As a consequence of our formula, the Virasoro constraints for the total descendant potential of $\mathbb{CP}^1$ can be expressed as equations on the Lax and wave operators. Similarly, we expect that the symmetries of the EBTH give Virasoro constraints for the total descendant potential of $\mathbb{CP}^1$ with two orbifold points
(cf.\ \cite{DZ99,Giv01,MT08,CvdL13}). 

Finally, \seref{s5} contains concluding remarks and open questions.

\section{The extended bigraded Toda hierarchy}\label{s2}
In this section, we review the extended Toda hierarchy from \cite{CDZ04} and its generalization,
the extended bigraded Toda hierarchy \cite{Car06}, following the approach of \cite{Tak10}. 
We discuss the Lax operator, wave operators, and wave functions of the hierarchy.

\subsection{Spaces of difference and differential-difference operators}\label{sdiffo}

We will consider functions of a variable $s$, and the shift operator $\La$ defined by $(\Lambda f)(s)	=	f(s+1)$.
The space $\A$ of (formal) \emph{difference operators} consists of all expressions of the form
\begin{equation*}
A=\sum_{i\in\Z}a_i(s)\Lambda^i
\end{equation*}
We have $\A=\A_+\oplus\A_-$ where $\A_+$ (respectively, $\A_-$) consists of $A\in\A$ such that $a_i=0$ for all $i<0$
(respectively, $i\ge0$). The projections of $A\in\A$ are defined by
\begin{equation*}
A_+	=\sum_{i\geq 0}a_i(s)\Lambda^i, \qquad
A_-	=	\sum_{i<0}a_i(s)\Lambda^i.
\end{equation*}

We let $\A_{++}$ be the space of $A\in\A$ such that $a_i=0$ for $i\ll0$ (i.e., the powers of $\La$ are bounded from below)
and $\A_{--}$ be the space of $A\in\A$ such that $a_i=0$ for $i\gg0$ (i.e., the powers of $\La$ are bounded from above).
Both $\A_{++}$ and $\A_{--}$ are associative algebras, where the product is defined by linearity and
\begin{equation*}
(a(s)\La^i )(b(s)\La^j) = a(s)b(s+i) \La^{i+j}.
\end{equation*}
We can multiply any $A\in\A$ by any element of $\A_\fin=\A_{++}\cap\A_{--}$.
However, in general, the product of an element of $\A_{++}$ with an element of $\A_{--}$ is not well defined.

We will also consider the space $\A[\ps]$ of (formal) \emph{differential-difference operators}, where $\La\ps=\ps\La$.
Note that such operators depend polynomially on $\ps$.
Again, there is a splitting $\A[\ps]=\A_+[\ps]\oplus\A_-[\ps]$, and we have the associative algebras
$\A_{++}[\ps]\supset\A_+[\ps]$ and $\A_{--}[\ps]\supset\A_-[\ps]$.
Differential-difference operators can be applied to $z^s$ and multiply it by a polynomial in $\log z$ whose coefficients are
formal power series in $z^{\pm1}$, according to the usual rule
\begin{equation*}
(a(s)\La^i \ps^j) z^s = a(s) z^i (\log z)^j  z^{s}.
\end{equation*}

\subsection{The extended Toda hierarchy}\label{seth}

The \emph{Lax operator} for the extended Toda hierarchy has the form
\begin{equation*}
			\label{LaxIntToda}
			L	=	\Lambda + v(s) + e^{u(s)}\Lambda\inv \in\A_\fin.
\end{equation*}
There are two \emph{wave operators} (also called ``dressing operators")
\begin{equation}\label{wavop}
\begin{split}
W	&=	1 + \sum_{i=1}^\infty w_i(s)\Lambda^{- i} \in 1+\A_- \subset \A_{--}, \\
 \W	&=	\sum_{i=0}^\infty \bar{w}_i(s)\Lambda^i \in\A_+, \qquad \bar{w}_0(s)	\neq	0 ,
\end{split}
\end{equation}
such that
\begin{equation*}
			L	=	W\Lambda W\inv	=	\W\Lambda\inv\W\inv.
\end{equation*}
The ``logarithm of $L$" is defined as follows \cite{CDZ04}:
\begin{equation}\label{logL}
			\log L	=	\frac{1}{2}W\p{s}W\inv - \frac{1}{2}\W\p{s}\W\inv	=	-\frac{1}{2}\frac{\partial W}{\partial s}W\inv + \frac{1}{2}\frac{\partial \W}{\partial s}\W\inv.
\end{equation}
Note that $\log L$ is a difference operator that commutes with $L$.

\begin{definition}\label{ExtTodaDef}
			The \emph{extended Toda hierarchy} (abbreviated ETH) in Lax form is given by:
\begin{equation}\label{ExtToda}
			\begin{aligned}	
				\p{t_n} L 	&=	[(L^n)_+,L],  &  \qquad & n\geq 1, \\
				\p{x_n} L	&=	[(2L^n\log L)_+,L],  &  \qquad & n\geq 0.
			\end{aligned}
\end{equation}
\end{definition}

\begin{remark}\label{reth1}
Our formulation of the ETH, which follows \cite{Tak10}, is equivalent to the original one from \cite{CDZ04} via the following change of variables:
		\begin{equation*}
			\begin{split}
				t^{1,k}	&=	\epsilon k! \, x_k,  \quad x=\epsilon s, \qquad\quad	k\geq 0,\\
				t^{2,k}	&=	\epsilon (k+1)! \, ( t_{k+1} + 2c_{k+1}x_{k+1} ),
			\end{split}	
		\end{equation*}
where 
\begin{equation}
			\label{harmnum}
			c_0	=	0,	\quad c_k	=	1 + \frac{1}{2} + \frac{1}{3} + \cdots + \frac{1}{k}
\end{equation}
		are the {harmonic numbers}. 
\end{remark}

The flows of the ETH act on the wave operators by:
		\begin{align*}
				\p{t_n}  W	
&=	-(L^n)_-W, \quad &
				\p{t_n}  \W	
&=	(L^n)_+\W, \\
				\p{x_n} W	&=	
-(2L^n\log L)_-W,\quad &
				\p{x_n} \W	&=	
(2L^n\log L)_+\W.
		\end{align*}
As in \cite{Tak10}, it is convenient to rewrite \eqref{ExtToda} as follows:
\begin{equation*}
			\begin{aligned}
				\p{t_n} L 	&=	[A_n,L],  &  \qquad & n\geq 1, \\
				\p{x_n} L	&=	[L^n\p{s} + P_n,L],  &  \qquad & n\geq 0,
			\end{aligned}
\end{equation*}
where
\begin{equation}\label{An}
A_n	=	\frac{1}{2}(L^n)_+ - \frac{1}{2}(L^n)_- ,\\
\end{equation}
and
\begin{equation}\label{Pn}
\begin{split}
				P_n	&=	-\Bigl(L^n\frac{\partial W}{\partial s}W\inv\Bigr)_+ - \Bigl(L^n\frac{\partial \W}{\partial s}\W\inv\Bigr)_-\\
						&=	L^nW\p{s}W\inv - (2L^n\log L)_- - L^n \p{s}\\
						&=	L^n\W\p{s}\W\inv + (2L^n\log L)_+ - L^n \p{s}.
\end{split}
\end{equation}

\begin{remark}\label{reth2}
We have $A_n,P_n\in\A_\fin$ and $P_0=0$. Hence, $\d_{x_0}-\d_s$ acts trivially on $W$, $\W$ and $L$,
and so they depend on $x_0+s$.
\end{remark}

Due to the above remark, from now on we will assume $x_0=s$.
Without writing it explicitly, we will consider $W$, $\W$ and $L$ as functions of $s$ and
\begin{equation*}
\t=(t_1,t_2,\dots), \qquad \x=(x_1,x_2,\dots).
\end{equation*}
Introduce the notation
\begin{equation}\label{xitz}
\xi(\t,z)=\sum_{i=1}^\infty t_iz^i
\end{equation}
		and
		\begin{equation*}
			\begin{split}
				\chi	=	z^{s + \xi(\x,z)}e^{\xi(\t,z)/2}, \qquad
				\bar{\chi}	 =	z^{s + \xi(\x,z\inv)}e^{-\xi(\t,z\inv)/2}.
			\end{split}
		\end{equation*}
Then the \emph{wave functions} of the ETH are defined by:
		\begin{equation*}
			\begin{split}
				\psi=\psi(s,\t,\x,z)				&=	W\chi					=	w\chi,\\
				\bar\psi=\bar{\psi}(s,\t,\x,z)	&=	\W \bar{\chi}	=	\bar{w}\bar{\chi},
			\end{split}
		\end{equation*}
		where
		\begin{equation}\label{wavfn}
\begin{split}
			w				
=	1	+	\sum_{i=1}^\infty w_i(s)z^{- i},\qquad
			\bar{w}	
=	\sum_{i=0}^\infty \bar{w}_i(s)z^i
\end{split}
		\end{equation}
are the (left) symbols of $W$ and $\W$, respectively.
By construction, we have:
		\begin{equation*}
			L\psi	=	z\psi,	\qquad	L\bar{\psi}	=	z\inv\bar{\psi},
		\end{equation*}
and
\begin{align*}
				\p{t_n}  \psi	&=	A_n\psi,	\quad &
				\p{t_n}  \bar\psi	&=	A_n\bar{\psi}, \\
				\p{x_n} \psi	&=	(L^n\p{s} + P_n)\psi,\quad &
				\p{x_n} \bar\psi	&=	(L^n\p{s} + P_n)\bar\psi.
\end{align*}

There exists a \emph{tau-function} $\tau(s,\t,\x)$ such that
		\begin{align*}
				w	&=	\frac\psi\chi = \frac{\tau(s,\t-[z\inv],\x)}{\tau(s,\t,\x)}		=	\frac{G(z)\tau(s,\t,\x)}{\tau(s,\t,\x)} \,,\\
				\bar w &= \frac{\bar\psi}{\bar\chi} 	=	\frac{\tau(s+1,\t+[z],\x)}{\tau(s,\t,\x)}	=	\frac{\bar{G}(z)\tau(s,\t,\x)}{\tau(s,\t,\x)} \,,
		\end{align*}
where
		\begin{equation}
			\label{zshift}
			[z]	=	\Bigl(z,\frac{z^2}{2},\frac{z^3}{3},\ldots\Bigr).
		\end{equation}
More explicitly,
		\begin{align*}
			G(z)	=	\exp\Bigl( -\sum_{i=1}^\infty \frac{z^{- i}}{i} \p{t_i} \Bigr), \qquad
	\bar{G}(z)	=	\exp\Bigl( \p{s} + \sum_{i=1}^\infty \frac{z^i}{i}\p{t_i} \Bigr).
\end{align*}
Observe that $\tau$ is unique only up to a multiplication by a function $f(s,\x)$ that satisfies
 $f(s+1,\x)=f(s,\x)$. The space of all such functions $f$ will be denoted by $\F$.

\subsection{The extended bigraded Toda hierarchy}\label{sebth}
{}From now on we will fix two positive integers $k$ and $m$.
Consider a \emph{Lax operator} of the form
		\begin{equation}
			\label{BTHLax}
			L	=	\Lambda^k + u_{k-1}(s)\Lambda^{k-1} + \cdots + u_{- m}(s)\Lambda^{- m} \in\A_\fin, \qquad u_{- m}\neq 0.
		\end{equation}
There exist \emph{wave operators} $W$ and $\W$ as in \eqref{wavop}, such that
		\begin{equation}
			\label{BTHLaxdress}
			L	=	W\Lambda^kW\inv	=	\W\Lambda^{- m}\W\inv.
		\end{equation}
This allows us to define fractional powers of $L$ for any integer $n$:
		\begin{equation}
			\label{EBTHfracL}
			L^{\frac{n}{k}}	=	W\Lambda^n W\inv	\in\A_{--},	\qquad	
L^{\frac{n}{m}}	=	\W\Lambda^{-n}\W\inv	\in\A_{++},
		\end{equation}
which commute with $L$ and satisfy
\begin{equation*}
\bigl(L^{\frac{n}{k}}\bigr)^k	=	\bigl( L^{\frac{n}{m}} \bigr)^m	=	L^n, \qquad n\in\Z_{\ge0}.
\end{equation*}
However, observe that $L^{\frac{n}{k}} \neq L^{\frac{p}{m}}$, unless $\frac{n}{k}=\frac{p}{m}\in\Z_{\ge0}$.
In particular, $L$ has two different inverses $L^{\frac{-k}{k}} \neq L^{\frac{-m}{m}}$, so we will not use
the notation $L\inv$.
As before, we define $\log L\in\A$ by \eqref{logL}. Then $\log L$ commutes with all $L^n$ for $n\in\Z_{\ge0}$,
but the composition of $\log L$ with a fractional power of $L$ is not well defined in general.

\begin{definition}\label{EBTHDef}
			The \emph{extended bigraded Toda hierarchy} (abbreviated EBTH) in Lax form is given by:
\begin{equation}\label{EBTHLaxflows}
			\begin{aligned}
\p{t_n} L &= [(L^{\frac{n}{k}})_+,L]	
, &	\qquad & n\geq 1, \\
\p{\bar t_n} L &= [(L^{\frac{n}{m}})_+,L]	
, &	\qquad & n\geq 1, \\
\p{x_n}	L &=	[(2L^n\log L)_+,L]	
,	&\qquad & n\geq 0.
			\end{aligned}
\end{equation}
		\end{definition}

The first two equations in \eqref{EBTHLaxflows} describe the \emph{bigraded Toda hierarchy},
which is a reduction of the 2D Toda hierarchy (see \cite{Tak10,UT84}). 
For $k=m=1$, the EBTH is equivalent to the extended Toda hierarchy.

The flows of the EBTH induce flows on the dressing operators:
\begin{equation}\label{EBTHdresflows}
		\begin{aligned}
\p{t_n} W	&=	-(L^{\frac{n}{k}})_-W,	& \qquad
\p{t_n} \W &=	(L^{\frac{n}{k}})_+\W,\\
\p{\bar{t}_n} W	&=	-(L^{\frac{n}{m}})_-W,	& \qquad \p{\bar{t}_n} \W	&=	(L^{\frac{n}{m}})_+\W, \\
			\p{x_n} W	&=	-(2L^n\log L)_-W,	& \qquad \p{x_n} \W	&=	(2L^n\log L)_+\W.
		\end{aligned}
\end{equation}

\begin{remark}\label{rebth2}
Since $\d_{x_0}-\d_s$ and $\p{t_{nk}}-\p{\bar{t}_{nm}}$ act trivially on $L$, $W$ and $\W$, it follows that 
$L$, $W$ and $\W$ depend on $x_0+s$ and $t_{nk}+\bar{t}_{nm}$ for $n\geq 1$. Without loss of generality, we will assume 
$x_0=s$ and $t_{nk}=\bar{t}_{nm}$, which amounts to making the substitution
$x_0+s \mapsto x_0$, $x_0-s \mapsto 0$, and
\begin{equation}\label{ttbar}
t_{nk}+\bar{t}_{nm} \mapsto t_{nk}, \quad t_{nk}-\bar{t}_{nm} \mapsto 0 \qquad (n\geq 1)
\end{equation}
in $L$, $W$ and $\W$.
\end{remark}

\begin{remark}\label{rebth1}
The version of the EBTH presented here is not the one originally introduced by Carlet in \cite{Car06} nor the one used in \cite{CvdL13}, but resembles more the ETH from \cite{Tak10} and \seref{seth}.
To compare it to the one from \cite{CvdL13}, first we need to change $\epsilon \mapsto-\epsilon$, which leads to 
$\Lambda \mapsto\Lambda\inv$ and $\zeta\mapsto\zeta\inv$ ($z$ here) in \cite{CvdL13}. Then we
have the following change of variables:
		\begin{align*}
				x								&=	\epsilon s,\\
				q^{k-\alpha}_n	&=	\epsilon k\Bigl( n + \frac{\alpha}{k} \Bigr)_{n+1}t_{nk+\alpha}, \quad \alpha=1,2,\ldots,k-1,\\
				q^{k+\beta}_n		&=	\epsilon m\Bigl( n + \frac{\beta}{m} \Bigr)_{n+1}\bar{t}_{nm+\beta}, \quad \beta=1,2,\ldots,m-1,\\
				q^{k+m}_n				&=	\epsilon m(n+1)!\Bigl( 
t_{(n+1)k}+\bar t_{(n+1)m} + c_{n+1}\Bigl(\frac{1}{k} + \frac{1}{m}\Bigr)x_{k+1} \Bigr),\\
				q^k_n						&=	\epsilon n! \, x_n,
		\end{align*}
where $n\geq 0$, $c_n$ are the harmonic numbers \eqref{harmnum}, and 
$(p)_n$ denotes the Pochhammer symbol,
		\begin{align*}
			(p)_0								&=	1,\\
			(p)_n						&=	\prod_{i=1}^n (p-i+1), \qquad n\geq 1, \\
			(p)_{-n}	&=	\prod_{i=-n+1}^0 (p-i+1)\inv	=	\frac{1}{(p+n)_n} \,.
		\end{align*}
The inverse of this change of variables can be used to translate our results to the version of the EBTH from \cite{CvdL13}.
\end{remark}

By \reref{rebth2}, from now on we will assume 
$x_0=s$ and $t_{nk}=\bar{t}_{nm}$ for $n\geq 1$. 
Introduce the notation
\begin{equation*}
\t=(t_1,t_2,\dots), \qquad \bt=(\bar t_1,\bar t_2,\dots), \qquad \x=(x_1,x_2,\dots),
\end{equation*}
and
\begin{equation*}
\xi_k(\t,z)=\sum_{i=1}^\infty t_{ki}z^{ki}.
\end{equation*}
We define the wave functions of the EBTH by:
		\begin{align*}
				\psi	=\psi(s,\t,\bt,\x,z)				& =	W\chi	=	w\chi,\\
				\bar{\psi}	=\bar\psi(s,\t,\bt,\x,z)	 & =	\W\bar{\chi}	=	\bar{w}\bar{\chi},
		\end{align*}
where
		\begin{align*}
			\chi				&=	z^{s + \xi(\x,z^k)}e^{\xi(\t,z) - \frac{1}{2}\xi_k(\t,z)},\\
			\bar{\chi}	&=	z^{s + \xi(\x,z^{-m})}e^{-\xi(\bt,z\inv) + \frac{1}{2}\xi_m(\bt,z\inv)},
		\end{align*}
and $\xi$, $w$, $\bar w$ are from \eqref{xitz}, \eqref{wavfn}.
Then
\begin{equation*}
L\psi=z^k\psi,  \qquad   L\bar\psi=z^{-m}\bar\psi,
\end{equation*}
and
\begin{equation}\label{psiflow}
		\begin{aligned}
\p{t_n}\psi &=	(L^{\frac{n}{k}})_+\psi, & \qquad n&\in\Z_{\geq 1} \setminus k\Z,\\
\p{\bar t_n}\psi &= -(L^{\frac{n}{m}})_-\psi, & \qquad n&\in\Z_{\geq 1} \setminus m\Z,\\
\p{t_{nk}}	\psi &= A_n\psi, & \qquad n&\in\Z_{\geq 1}, \\
\p{x_n}\psi &=	(L^n\p{s} + P_n)\psi, & \qquad n&\in\Z_{\geq 0},
		\end{aligned}
\end{equation}
where $A_n$ and $P_n$ are given by \eqref{An}, \eqref{Pn}. 
Exactly the same equations hold with $\psi$ replaced by $\bar\psi$.
Observe that, due to \eqref{EBTHfracL}, we have
$(L^{\frac{n}{k}})_+ \in\A_\fin$ and $(L^{\frac{n}{m}})_- \in\A_\fin$.

\section{Action of additional symmetries on the Lax and wave operators}\label{s3}

In this section, we find two operators $M$ and $\M$ such that $[L,M]=[L,\M]=1$, and use them to construct additional symmetries of the extended bigraded Toda hierarchy. We show that the symmetries corresponding to $L^p(M-\M)$ commute as the Virasoro 
vector fields $-z^p \p{z}$ ($p\ge0$).
We will continue using the notation from \seref{sebth}; in particular, we will 
assume $x_0=s$ and $t_{nk}=\bar{t}_{nm}$ for $n\geq 1$.

\subsection{Vector fields and their commutators}\label{svfc}

For two fixed positive integers $k$ and $m$,
consider the (infinite-dimensional) manifold $\MM$ consisting of all triples $(W,\W,L)$ satisfying \eqref{wavop}, \eqref{BTHLaxdress}.
%
Note that $\MM$ is a submanifold of the affine space $(1+\A_-)\times\A_+\times\A_\fin$; hence the tangent spaces of $\MM$ can be embedded
as subspaces of the vector space $\A_-\oplus\A_+\oplus\A_\fin$. Thus, vector fields on $\MM$ can be viewed as functions
from $\MM$ to $\A_-\oplus\A_+\oplus\A_\fin$.

\begin{definition}\label{VDef}
Let $\V$ be the space of functions $a\colon\MM\to\A$ with the property that $[a(X),L]=0$ for $X=(W,\W,L)\in\MM$.
For $a\in\V$, we define 
\begin{equation}\label{dda}
\dd_a(X)=\bigl(-a(X)_-W,a(X)_+\W,[a(X)_+,L]\bigr). 
\end{equation}
For brevity, when $X$ is fixed, we will write just $a$ instead of $a(X)$.
\end{definition}

The right-hand side of \eqref{dda} is an element of the tangent space of $\MM$ at the point $X$, because
\begin{align*}
\dd_a(W\Lambda^kW\inv) &= (\dd_a W)\Lambda^kW\inv - W\Lambda^kW\inv (\dd_a W) W\inv \\
&= [-a_-,W\Lambda^kW\inv] =[-a_-,L] \\
&=[a_+,L] = [a_+,\W\Lambda^{- m}\W\inv] \\
&=\dd_a(\W\Lambda^{- m}\W\inv).
\end{align*}
Therefore, every $\dd_a$ ($a\in\V$) is a vector field on $\MM$. 
For example, by \eqref{EBTHLaxflows} and \eqref{EBTHdresflows}, the flows of the EBTH correspond to the vector fields
\begin{equation}\label{vflows}
\p{t_n} =	\dd_{L^{\frac{n}{k}}},  \qquad
\p{\bar{t}_n} =	\dd_{L^{\frac{n}{m}}}, \qquad
\p{x_n} =	\dd_{2L^n\log L} \quad\text{on \;$\MM$.}
\end{equation}
The following key lemma calculates the brackets of the vector fields $\dd_a$
(cf.\ \cite[Lemma 2.1]{ASvM95}).

\begin{lemma}\label{ddabLem1}
For\/ $a,b\in\V$, suppose that at every point\/ $X\in\MM$ we can write\/
$a=A-\bar A$ and\/ $b=B-\bar B$, where
\begin{equation*}
A,B\in\A_{--}+\A_{--}\ps, \qquad \bar A,\bar B\in\A_{++}+\A_{++}\ps
\end{equation*}
are such that
\begin{equation*}
\begin{aligned}
\dd_a(b_\pm)&=(\dd_a b)_\pm, &  \quad 
\dd_a B&=[-a_-,B], &  \quad
\dd_a\bar B&=[a_+,\bar B], \\
\dd_b(a_\pm)&=(\dd_b a)_\pm, & \quad 
\dd_b A &=[-b_-,A], &  \quad
\dd_b \bar A&=[b_+,\bar A].
\end{aligned}
\end{equation*}
Then\/ $[\dd_a,\dd_b]=\dd_c$ 
where\/ $c=[A,B]-[\bar A,\bar B] \in\V$. 
\end{lemma}
\begin{proof}
Fix a point $X=(W,\W,L)\in\MM$.
First, the same calculation as in the proof of \cite[Lemma 2.1]{ASvM95} gives
\begin{equation*}
[\dd_a,\dd_b](W) = -d_- W, \qquad
d=-\dd_a b+\dd_b a-[a_-,b_-].
\end{equation*}
We find
\begin{align*}
-\dd_a b \;+\; & \dd_b a 
=	[a_- ,B] + [a_+,\bar{B}] - [b_-,A] - [b_+,\bar{A}] \\
=\; & [A_-,B] - [\bar{A}_-,B] + [A_+,\bar{B}] - [\bar{A}_+,\bar{B}] \\
&- [B_-,A] + [\bar{B}_-,A] - [B_+,\bar{A}] + [\bar{B}_+,\bar{A}]\\
=\; & [A,B] - [A_+,B_+] - [A_+,B_-] - [\bar{A}_-,B_+] - [\bar{A}_-,B_-] \\
&+ [A_+,\bar{B}_+] + [A_+,\bar{B}_-] 
- [\bar{A}_+,\bar{B}_+] - [\bar{A}_+,\bar{B}_-]\\
&- [B_-,A_+] - [B_-,A_-] + [\bar{B}_-,A_+] + [\bar{B}_-,A_-] \\
&- [B_+,\bar{A}_+] - [B_+,\bar{A}_-] + 
[\bar{B},\bar{A}] - [\bar{B}_-,\bar{A}_+] - [\bar{B}_-,\bar{A}_-] \\
=\; & c - [a_+,b_+] + [a_-,b_-],
\end{align*}
where
\begin{equation*}
c = [A,B]-[\bar A,\bar B].
\end{equation*}
This implies $d_-=c_-$.

In the same way, 
\begin{equation*}
[\dd_a,\dd_b](\W)= e_+\W, \qquad
e=-\dd_a b+\dd_b a+[a_+,b_+],
\end{equation*}
and $e_+=c_+$.
Similarly, we also have
\begin{align*}
[\dd_a,\dd_b](L) = [-c_-,L] = [c_+,L],
\end{align*}
which implies that $[c,L]=0$ at every point $X=(W,\W,L)\in\MM$.

Finally, it remains to check that $c\in\A$. 
Let us define
\begin{equation*}
\tilde a = -a_- + A = a_+ + \bar A \in \A_{--}[\ps]\cap\A_{++}[\ps] = \A_\fin[\ps],
\end{equation*}
and similarly for $\tilde b$. Then 
\begin{align*}
[\tilde a,b]+[a,\tilde b] &= [-a_- + A,b_-] + [a_+ + \bar A,b_+] + [A,-b_- + B] - [\bar A,b_+ + \bar B] \\
&= c + [a_+,b_+] - [a_-,b_-].
\end{align*}
Hence $c\in\A$,
because $a,b\in\A$ and $\tilde a,\tilde b$ have degree $\le1$ in $\ps$.
This completes the proof of the lemma.
\end{proof}

\leref{ddabLem1} is applicable to the vector fields \eqref{vflows} and confirms that they commute among themselves.
In more detail, in each case we can write $a=A-\bar A$ as follows $(n\geq 0)$:
\begin{equation}\label{vflows2}
\begin{aligned}
a&=L^{\frac{n}{k}}, &\;\; A&=L^{\frac{n}{k}}, &\;\; \bar A&=0,\\
a&=L^{\frac{n}{m}},  &\;\; A&=0, &\;\; \bar A&=-L^{\frac{n}{m}},\\
a&=2L^n\log L, &\;\; A&=L^nW\p{s}W\inv, &\;\; \bar A&=L^n\W\p{s}\W\inv.
\end{aligned}
\end{equation}
Note that $a=L^n$ can be split as $\al L^n-(\al-1)L^n$ for any $\al\in\CC$, and they all give rise to the same result.

Let $\E$ be the (infinite-dimensional) manifold consisting of all quintuples $Y=(W,\W,L,\psi,\bar\psi)$ such that $L$ is a solution of the EBTH (for fixed $k,m$) with wave operators $W,\W$ and wave functions $\psi,\bar\psi$, as in \seref{sebth}.
We extend \eqref{dda} to $\E$ by letting
\begin{equation}\label{dda2}
\dd_a(Y)=\bigl(-a_-W,a_+\W,[a_+,L],-a_-\psi,a_+ \bar\psi \bigr), \qquad Y\in\E,
\end{equation}
where as before, we write $a$ instead of $a(X)$ for $X=(W,\W,L)$.
However, the flows of the EBTH are no longer given by \eqref{vflows} and do not have the form \eqref{dda2}, because by \eqref{psiflow} their action on $\psi$ and $\bar\psi$ is different
(cf.\ \cite[Remark 2.1.2]{ASvM95} or \cite[Remark 7.1.8]{Di03}).
Instead, they can be described as follows.

\begin{definition}\label{VDef2}
Let $\VC$ be the space of pairs of functions $(A,\bar A)$, where 
\begin{equation*}
A\colon\MM\to\A_{--}+\A_{--}\ps, \qquad
\bar A\colon\MM\to\A_{++}+\A_{++}\ps
\end{equation*}
are such that $a=A-\bar A\in\A$ and $[A,L]=[\bar A,L]=0$ at every $X=(W,\W,L)\in\MM$
(as before, we write $A$ instead of $A(X)$, etc.).
Then we define 
\begin{equation}\label{dda3}
\dc_{A,\bar A}(Y)=\bigl(-a_-W,a_+\W,[a_+,L],\tilde a\psi,\tilde a\bar\psi \bigr),
\end{equation}
where
\begin{equation}\label{aabar}
\tilde a = A_+ + \bar A_- = -a_- + A = a_+ + \bar A \in \A_\fin+\A_\fin\ps,
\end{equation}
for $(A,\bar A)\in\VC$ and $Y=(W,\W,L,\psi,\bar\psi)\in\E$.
\end{definition}

\begin{proposition}\label{pdcebth}
The flows of the EBTH can be written in the form\/ $\dc_{A,\bar A}$ as follows$:$
\begin{equation*}
\begin{aligned}
\p{t_n} &\colon&\;\;  A&=L^{\frac{n}{k}}, &\;\; \bar A&=0 & \;\; (n&\in\Z_{\geq 1} \setminus k\Z),\\
\p{\bar t_n} &\colon &\;\; A&=0, &\;\; \bar A&=-L^{\frac{n}{m}}& \;\; (n&\in\Z_{\geq 1} \setminus m\Z),\\
\p{t_{nk}} &\colon&\;\;  A&=\frac12 L^n, &\;\; \bar A&=-\frac12 L^n & \;\; (n&\in\Z_{\geq 1}), \\
\p{x_n} &\colon &\;\; A&=L^nW\p{s}W\inv, &\;\; \bar A&=L^n\W\p{s}\W\inv & \;\; (n&\in\Z_{\geq 0}).
\end{aligned}
\end{equation*}
\end{proposition}
\begin{proof}
This follows immediately from \eqref{EBTHLaxflows}, \eqref{EBTHdresflows}, \eqref{psiflow},
by using \eqref{An} and \eqref{Pn}.
\end{proof}

\begin{remark}\label{rdcebth2}
Clearly, $a=A-\bar A \in\V$ for every $(A,\bar A)\in\VC$.
By extending the proof of \leref{ddabLem1} to include the actions on $\psi$ and $\bar\psi$, it is easy to verify
that the flows of the EBTH indeed commute among themselves as vector fields on $\E$.
\end{remark}

Under suitable assumptions, in the next lemma we prove that 
the vector fields $\dc_{A,\bar A}$ commute with $\dd_b$
(cf.\ \cite[Lemma 2.1]{ASvM95}).

\begin{lemma}\label{ddabLem2}
Let\/ $(A,\bar A)\in\VC$ and\/ $b\in\V$ be such that 
\begin{equation*}
\begin{aligned}
\dc_{A,\bar A}(b_\pm)&=(\dc_{A,\bar A} b)_\pm,  &\quad   
\dd_b(A_\pm)&=(\dd_b A)_\pm, &  \quad
\dd_b(\bar A_\pm)&=(\dd_b \bar A)_\pm, \\
\dc_{A,\bar A} b&=[\tilde a,b], &\quad
\dd_b A &=[-b_-,A], &\quad
\dd_b \bar A&=[b_+,\bar A],
\end{aligned}
\end{equation*}
where\/ $\tilde a = A_+ + \bar A_-$.
Then\/ $[\dc_{A,\bar A},\dd_b]=0$ on\/ $\E$.
\end{lemma}
\begin{proof}
The proof is similar to that of \leref{ddabLem1}, so some details will be omitted. 
First, we find the action of $[\dc_{A,\bar A},\dd_b]$ on $W$:
\begin{equation*}
[\dc_{A,\bar A},\dd_b](W) = -c_-W, \qquad c = -\dc_{A,\bar A} b+\dd_b a-[a_-,b_-].
\end{equation*}
We calculate:
\begin{align*}
\dc_{A,\bar A} b &- \dd_b a =	[A_+ + \bar A_-,b] - [-b_-,A] - [b_+,-\bar{A}] \\
&= [A_+,b_+] + [A_+,b_-] + [\bar{A}_-,b_+] + [\bar{A}_-,b_-] - [A,b_-] - [\bar{A},b_+] \\
&= [A_+,b_+] -  [A_-,b_-] -  [\bar{A}_+,b_+] + [\bar{A}_-,b_-] \\
&= [a_+,b_+] -  [a_-,b_-],
\end{align*}
which implies that $c_-=0$. Similarly,
\begin{equation*}
[\dc_{A,\bar A},\dd_b](\W) = d_+\W, \qquad d = -\dc_{A,\bar A} b+\dd_b a+[a_+,b_+],
\end{equation*}
and $d_+=0$. In the same way, $[\dc_{A,\bar A},\dd_b](L)=0$.

Next, we find
$
[\dc_{A,\bar A},\dd_b](\psi) = -e\psi,
$
where
\begin{align*}
e &= -\dc_{A,\bar A} (b_-) - \dd_b(\tilde a) + [\tilde a,b_-] \\
&= -(\dc_{A,\bar A} b)_-  - \dd_b(-a_- + A) + [-a_- + A,b_-] \\
&= c_- - \dd_b A + [A,b_-] = 0,
\end{align*}
using \eqref{aabar}. Similarly,
$
[\dc_{A,\bar A},\dd_b](\bar\psi) = f\bar\psi,
$
with
\begin{align*}
f &= -(\dc_{A,\bar A} b)_+  + \dd_b(a_+ + \bar A) + [a_+ + \bar A,b_+] \\
&= d_+ + \dd_b \bar A + [\bar A,b_+] = 0.
\end{align*}
This completes the proof.
\end{proof}

Lemmas \ref{ddabLem1} and \ref{ddabLem2} will be applied in  \seref{ssymebth} below to construct additional symmetries of the EBTH and to determine their commutators.

\subsection{The operators $M$ and $\M$}\label{smmbar}

First, by using that
\begin{align*}
			\Lambda^n \chi	&=	z^n\chi,		&		\Lambda^n \bar{\chi}	&=	z^{n}\bar{\chi} \qquad (n\in\Z),\\
			\p{s}\chi	&=	\log(z)\chi,			&		\p{s} \bar{\chi}	&=	\log(z)\bar{\chi},
		\end{align*}
we derive differential-difference operators $\Gamma,\bar\Ga\in\A[\ps]$ such that 
$\Gamma \chi = \p{z^k}\chi$ and $\bar\Ga\bar\chi=\p{z^{-m}}\bar\chi$.
We calculate:
\begin{align*}
\frac{k z^k \p{z^k} \chi}\chi &= z \p{z} \log(\chi) \\
&= z \p{z} \Bigl(\log(z)(s + \xi(\x,z^k)) + \xi(\t,z) - \frac{1}{2}\xi_k(\t,z) \Bigr) \\
& =s + \sum_{n= 1}^\infty
\Bigl( x_n z^{nk}+ nkx_n z^{nk}\log(z) + nt_n z^{n} - \frac{nk}{2} t_{nk} z^{nk} \Bigr),
		\end{align*}
which gives $\p{z^k}\chi=\Gamma \chi$ with
\begin{equation*}
\Gamma	= \frac{s}{k}\Lambda^{-k}  + \sum_{n= 1}^\infty
\Bigl( \frac1k x_n \La^{(n-1)k} + nx_n \La^{(n-1)k} \ps + \frac{n}k t_n \La^{n-k} - \frac{n}{2} t_{nk} \La^{(n-1)k} \Bigr) .
		\end{equation*}
By a similar calculation,
\begin{align*}
\frac{-m z^{-m} \p{z^{-m}} \bar\chi} {\bar\chi}  
&= z \p{z} \log(\bar\chi) \\
=s + \sum_{n= 1}^\infty
\Bigl( & x_n z^{-nm} - nmx_n z^{-nm}\log(z) + n\bar t_n z^{-n} - \frac{nm}{2} \bar t_{nm} z^{-nm} \Bigr).
		\end{align*}
Hence, $\p{z^{-m}}\bar\chi=\bar\Ga\bar\chi$ where
		\begin{equation*}
\bar\Ga = -\frac{s}{m}\Lambda^{m}  - \sum_{n= 1}^\infty
\Bigl( \frac1m x_n \La^{(1-n)m} - nx_n \La^{(1-n)m} \ps + \frac{n}m \bar t_n \La^{m-n} - \frac{n}{2} \bar t_{nm} \La^{(1-n)m} \Bigr) .
		\end{equation*}

\begin{definition}\label{MMbarDef}
With the above notation, we define $M=W\Ga W\inv$ and $\M=\W\bar\Ga\W\inv$.
		\end{definition}

Explicitly, by \eqref{BTHLaxdress}, \eqref{logL} and the above formulas for $\Ga,\bar\Ga$, we get:
		\begin{align*}
			M	=&	\; \frac{1}{k}Ws\Lambda^{- k} W\inv \\
 &+ \sum_{n=1}^\infty \Bigl( \frac{1}{k}x_nL^{n-1} + nx_n L^{n-1} W\p{s}W\inv  + \frac{n}{k} t_nL^{\frac{n-k}{k}} - \frac{n}{2}t_{nk}L^{n-1} \Bigr), \\
			\M	=&	-\frac{1}{m}\W s\Lambda^m \W\inv \\
 &- \sum_{n=1}^\infty \Bigl( \frac{1}{m}x_n L^{n-1} - nx_n L^{n-1}\W\p{s}\W\inv + \frac{n}{m}\bar{t}_n L^{\frac{n-m}{m}} - \frac{n}{2}\bar{t}_{nm}L^{n-1} \Bigr),
\end{align*}
and
\begin{equation}\label{EBTHMmMbar}
\begin{split}
				M - \M	=&\; \frac{1}{k}Ws\Lambda^{- k} W\inv + \frac{1}{m}\W s\Lambda^m \W\inv \\
&+ \sum_{n=1}^\infty \Bigl( \Bigl(\frac{1}{k} + \frac{1}{m}\Bigr)x_n + 2 nx_n\log L - n t_{nk} \Bigr) L^{n-1} \\
 &+ \sum_{n=1}^\infty \Bigl( \frac{n}{k}t_n L^{\frac{n-k}{k}} + \frac{n}{m}\bar{t}_nL^{\frac{n-m}{m}} \Bigr),
\end{split}
\end{equation}
where, as before, we assume $t_{nk}=\bar{t}_{nm}$.

\begin{example}\label{MMbarRem}
In the case of the ETH, when $k=m=1$, the above expression simplifies to:
\begin{align*}
M - \M	= Ws\Lambda^{- 1} W\inv + \W s\Lambda \W\inv
+ \sum_{n=1}^\infty \bigl( 2x_n + 2 nx_n\log L + nt_{n}\bigr) L^{n-1} .
\end{align*}
\end{example}

Note that $M$ and $\M$ are differential-difference operators, while $M-\M$ is a difference operator.
The infinite sums are well defined if we assume that all but finitely many of the variables are set to $0$, or more generally, they
converge in the topology given by $x_n,t_n,\bar t_n\to 0$ for $n\to+\infty$.

\begin{remark}\label{rmmbar} 
If we set all $x_n=0$ for $n\ge1$, our operators $M$ and $\M$ coincide with the operators $M_L$ and $M_R$ from \cite[Eqs.\ (4.1)--(4.3)]{LHS12} (up to a relabeling of the remaining variables).
An explicit calculation of the coefficient of $\La^0$ in $(M-\M)^2$ shows that this operator is not well defined even after setting $x_n=t_n=\bar{t}_n=0$ for $n\ge1$.
\end{remark}

\begin{lemma}\label{MMbarLem1}
The operators\/ $M$ and\/ $\M$ satisfy
\begin{equation}\label{MMpsi}
M\psi	=	\p{z^k}\psi,	\qquad	\M\bar{\psi}	=	\p{z^{- m}}\bar{\psi},
\end{equation}
and\/ $[L,M]	=	[L,\M]	=	1$. In particular, $[L,M-\M]	=0$.
\end{lemma}
\begin{proof}
Equation \eqref{MMpsi} holds by construction. The formulas for the commutators follow from \eqref{BTHLaxdress} and the observation that $L$ commutes with the sums in $M$ and $\M$.
\end{proof}

\begin{proposition}\label{pMMbarf}
The flows of the EBTH act on the operator\/ $M$ as follows$:$
\begin{equation}\label{Mflows}
\begin{aligned}
\p{t_n}M &=	[(L^{\frac{n}{k}})_+,M], & \qquad n&\in\Z_{\geq 1} \setminus k\Z,\\
\p{\bar t_n}M &= [-(L^{\frac{n}{m}})_-,M], & \qquad n&\in\Z_{\geq 1} \setminus m\Z,\\
\p{t_{nk}}	M &= [A_n,M], & \qquad n&\in\Z_{\geq 1}, \\
\p{x_n}M &=	[L^n\p{s} + P_n,M], & \qquad n&\in\Z_{\geq 0},
\end{aligned}
\end{equation}
where, as before, $A_n$ and\/ $P_n$ are given by \eqref{An}, \eqref{Pn}. 
The same equations hold with\/ $M$ replaced by\/ $\M$, if we assume\/ 
$t_{nk}=\bar{t}_{nm}$ for\/ $n\geq 1$.
\end{proposition}
\begin{proof}
Let $n\in\Z_{\geq 1} \setminus k\Z$.
It follows from \eqref{EBTHdresflows} that
\begin{equation*}
\begin{split}
\p{t_n}(WSW\inv) &= \bigl(-(L^{\frac{n}{k}})_-W\bigr)SW\inv - WSW\inv\bigl(-(L^{\frac{n}{k}})_-W\bigr) W\inv \\
&= [-(L^{\frac{n}{k}})_-,WSW\inv]
\end{split}
\end{equation*}
for any operator $S\in\A_{--}[\ps]$ such that $\p{t_n}S=0$. Thus
\begin{align*}
\p{t_n}M &= [-(L^{\frac{n}{k}})_-,M] + \frac{n}{k} L^{\frac{n-k}{k}} 
= [-(L^{\frac{n}{k}})_-,M] + [L^{\frac{n}{k}},M] \\
&=	[(L^{\frac{n}{k}})_+,M],
\end{align*}
using that $[L,M]=1$. The other equations in \eqref{Mflows} are proved in the same way.
\end{proof}

\subsection{Additional symmetries of the EBTH}\label{ssymebth}

Recall the (infinite dimensional) manifold $\E$, which consists of all $Y=(W,\W,L,\psi,\bar\psi)$ such that $L$ is a solution of the EBTH (for fixed $k,m$) with wave operators $W,\W$ and wave functions $\psi,\bar\psi$. Consider the following functions from $\E$ to $\A$:
\begin{equation*}
L^{\frac{n}{k}}, \;\; L^{\frac{n}{m}}, \;\; 2L^p\log L \qquad (n\in\Z,\; p\in\Z_{\ge0}),
\end{equation*}
and denote their linear span by $\V_0$. We also have the functions that send $Y\in\E$ to
\begin{equation*}
L^p(M-\M) \qquad (p\in\Z_{\ge0}),
\end{equation*}
and their linear span will be denoted by $\V_1$. 
We have $\V_0+\V_1\subset\V$, since $L$ commutes with $M-\M$, $\log L$ and all powers of $L$. Then every $a\in\V_0+\V_1$ gives rise to a vector field $\dd_a$ on $\E$ defined by \eqref{dda2}. Now we can formulate our first main result.

\begin{theorem}\label{thm1}
The vector fields\/ $\dd_a$ $(a\in\V_0+\V_1)$ on\/ $\E$ commute with the flows 
of the extended bigraded Toda hierarchy$:$
\begin{equation}\label{commfl}
[\p{y},\dd_a]=0, \qquad y\in\{ x_{n-1},t_n,\bar t_n \,|\, n\ge1 \}, \;\; a\in\V_0+\V_1,
\end{equation}
i.e., they are symmetries of the hierarchy.
Furthermore, we have$:$
\begin{align}
\label{comm1}
\bigl[ \dd_{L^p(M-\M)}, \dd_{L^{\frac{n}{k}}} \bigr] &= -\frac{n}{k} \dd_{L^{p+\frac{n-k}{k}}}, \\
\label{comm2}
\bigl[ \dd_{L^p(M-\M)}, \dd_{L^{\frac{n}{m}}} \bigr] &= -\frac{n}{m} \dd_{L^{p+\frac{n-m}{m}}}, \\
\label{comm3}
\bigl[ \dd_{L^p(M-\M)}, \dd_{2L^q\log L} \bigr] &= 
-q \dd_{2L^{p+q-1} \log L} 
- \dd_{ L^{p+q} \bigl( \frac1k L^{\frac{-k}k} + \frac1m L^{\frac{-m}m} \bigr) }, \\
\label{comm4}
\bigl[ \dd_{L^p(M-\M)}, \dd_{L^q(M-\M)} \bigr] &= 
(p-q)\dd_{L^{p+q-1}(M-\M)}, 
\end{align}
for\/ $n\in\Z$, $p,q\in\Z_{\ge0}$, and\/ $[\dd_a,\dd_b]=0$ for\/ $a,b\in\V_0$.
\end{theorem}
\begin{proof}
To prove \eqref{commfl}, recall that the flows of the EBTH can be written in the form 
$\dc_{A,\bar A}$ with $(A,\bar A)\in\VC$ (see \prref{pdcebth}). 
The conditions of \leref{ddabLem2} hold for them and for $b\in\V_0+\V_1$, by \prref{pMMbarf}. Therefore, $[\dc_{A,\bar A},\dd_b]=0$.

The commutators \eqref{comm1}--\eqref{comm4} can be derived from \leref{ddabLem1}
using that $[L,M]	=	[L,\M]	=	1$ (see \leref{MMbarLem1}).
For $a\in\V_0$, the splitting $a=A-\bar A$ is given in \prref{pdcebth}, and for
$a=L^p(M-\M)$, we have $A=L^p M$, $\bar A=L^p \M$.
Let us check for example \eqref{comm1}. In this case, we have:
\begin{equation*}
A=L^p M, \quad \bar A=L^p \M, \quad B=L^{\frac{n}{k}} = W\La^n W\inv, \quad \bar B=0.
\end{equation*}
Then
\begin{align*}
[A,B]&=L^p [M,L^{\frac{n}{k}}] = L^p W [\Ga,\La^n] W\inv \\
&= -\frac{n}{k} L^p W \La^{n-k} W\inv = -\frac{n}{k} L^{p+\frac{n-k}{k}}.
\end{align*}
Similarly, to prove \eqref{comm3} consider
\begin{equation*}
A=L^p M, \quad \bar A=L^p \M, \quad B=L^q W\ps W\inv, \quad \bar B=L^q\W\ps\W\inv.
\end{equation*}
Then
\begin{align*}
[A,B]&=L^p [M,L^q] W\ps W\inv + L^{p+q} [M,W\ps W\inv] \\
&= -q L^{p+q-1} W\ps W\inv + L^{p+q} W [\Ga,\ps] W\inv \\
&= -q L^{p+q-1} W\ps W\inv - \frac1k L^{p+q} W \La^{-k} W\inv \\
\intertext{and}
[\bar A,\bar B]&=-q L^{p+q-1} \W\ps \W\inv + \frac1m L^{p+q} \W \La^m \W\inv.
\end{align*}
The other two equations, \eqref{comm2} and \eqref{comm4}, are proved in the same way.
\end{proof}

It follows from \eqref{comm4} that $\V_1$ is a Lie algebra isomorphic to the Lie algebra
$\Der\CC[z]$ of vector fields on the line, under the map $\dd_{L^p(M-\M)} \mapsto -z^p\p{z}$.
By \eqref{comm1}--\eqref{comm3}, $\V_0$ is a module of $\V_1$ under the adjoint action.

\begin{remark}\label{remtn}
Although $\dd_{L^{\frac{n}{k}}}$ acts as $\p{t_n}$ on the Lax operator $L$ and wave operators $W$, $\W$, it 
does not commute with $\dd_{L^p(M-\M)}$, while $[\p{t_n},\dd_{L^p(M-\M)}]=0$. The reason is that
$\dd_{L^{\frac{n}{k}}}$ differs from $\p{t_n}$ when acting on $\psi$, $\bar\psi$ and $M$, $\M$
 (cf.\ \cite[Remark 2.1.2]{ASvM95} or \cite[Remark 7.1.8]{Di03}).
\end{remark}

\section{Action of additional symmetries on the tau-function}\label{s4}

Here we determine how the additional symmetries of the EBTH from the previous section act on the tau-function,
thus getting an analog of the Adler--Shiota--van Moerbeke formula (cf.\ \cite{ASvM95, AvM95}).
We show that the symmetries $\dd_a$ ($a\in\V_0$) give rise to an action of the Heisenberg algebra and are essentially trivial,
while $\dd_a$ ($a\in\V_1$) correspond to a subalgebra of the Virasoro algebra. 

\subsection{Tau-function of the EBTH}\label{stauebth}

Due to \cite{Car06,CvdL13},
there exists a \emph{tau-function} $\tau=\tau(s,\t,\bt,\x)$ such that
\begin{equation}\label{EBTHtaudress}
\begin{split}
w	&=	\frac\psi\chi = \frac{\tau(s,\t-[z\inv],\bt,\x)}{\tau(s,\t,\bt,\x)}		=	\frac{G(z)\tau}{\tau} \,,\\
\bar w &= \frac{\bar\psi}{\bar\chi} 	=	\frac{\tau(s+1,\t,\bt+[z],\x)}{\tau(s,\t,\bt,\x)}	=	\frac{\bar{G}(z)\tau}{\tau} \,,
\end{split}
\end{equation}
where $[z]$ is as usual given by \eqref{zshift}. 
Explicitly,
\begin{equation}\label{GGbar}
			G(z)	=	\exp\Bigl( -\sum_{i=1}^\infty \frac{z^{- i}}{i} \p{t_i} \Bigr), \qquad
	\bar{G}(z)	=	\exp\Bigl( \p{s} + \sum_{i=1}^\infty \frac{z^i}{i}\p{\bar t_i} \Bigr).
\end{equation}

In this section, as before, we will continue to assume that $x_0=s$.
However, in contrast with the previous sections where we assumed $t_{nk}=\bar{t}_{nm}$ ($n\ge1$), we will now treat these variables as distinct (the reason will become clear later in \seref{sheis}).
Then for every fixed $\x$, $\tau$ is just the tau-function of the 2D Toda hierarchy corresponding to the pair of Lax operators $\mathcal{L}=W\La W\inv$, $\bar{\mathcal{L}}=\W\La\W\inv$, and \eqref{BTHLaxdress} gives the
reduction $L=\mathcal{L}^k = \bar{\mathcal{L}}^{-m}$ (see \cite{Tak10,UT84} and \reref{rebth2}).
Note that for $k=m=1$, if $\tau$ depends on the sums $t_{n}+\bar{t}_{n}$,
the operators $G(z)$ and $\bar G(z)$ coincide with the ones from
\seref{seth}.
As for the ETH, $\tau$ is determined only up to a multiplication by a function $f(s,\x)\in\F$,
i.e., such that $f(s+1,\x)=f(s,\x)$.

Recall from \reref{rebth2} that 
$L$, $W$ and $\W$ depend on $t_{nk}+\bar{t}_{nm}$ for $n\geq 1$, and we made the substitution \eqref{ttbar} in them.
Therefore, since we now treat the variables $t_{nk}$ and $\bar{t}_{nm}$ as distinct, we have to make the inverse substitution
\begin{equation*}
t_{nk} \mapsto t_{nk}+\bar{t}_{nm}, \quad \bar{t}_{nm} \mapsto t_{nk}+\bar{t}_{nm} \qquad (n\geq 1)
\end{equation*}
in $M-\M$ (otherwise, $M-\M$ would not satisfy \eqref{Mflows}). Thus, \eqref{EBTHMmMbar} is modified as follows:
\begin{equation}\label{MMbar2}
\begin{split}
				M - \M	=&\; \frac{1}{k}Ws\Lambda^{- k} W\inv + \frac{1}{m}\W s\Lambda^m \W\inv \\
&+ \sum_{n=1}^\infty \Bigl( \Bigl(\frac{1}{k} + \frac{1}{m}\Bigr)x_n + 2 nx_n\log L \Bigr) L^{n-1} \\
 &+ \sum_{n=1}^\infty \Bigl( \frac{n}{k}t_n L^{\frac{n-k}{k}} + \frac{n}{m}\bar{t}_nL^{\frac{n-m}{m}} \Bigr),
\end{split}
\end{equation}
Then all results from the previous section remain true.

\subsection{Vector fields on the manifold of tau-functions}\label{svectau}

For fixed positive integers $k,m$, denote by $\T$ the (infinite-dimensional) manifold of EBTH tau-functions, and recall that $\E$
is the manifold of EBTH Lax operators, wave operators and wave functions (see \seref{svfc}).
Also recall that for $a\in\V_0+\V_1$, we have the vector field $\dd_a$ on $\E$ given by \eqref{dda2} (see \seref{ssymebth}).
These vector fields induce vector fields on $\T$, and
the correspondence is given by the following \emph{Adler--Shiota--van Moerbeke formula}
(cf.\ \cite[Theorem 5.1]{vM91} and \cite[Lemma 3.1]{AvM95}).

\begin{proposition}\label{pvftau}
The vector fields\/ $\dd_a$ $(a\in\V_0+\V_1)$ on\/ $\E$ induce vector fields\/ $\dh_a$ on\/ $\T$ so that
\begin{equation}\label{vftau}
\begin{split}
\frac{\dd_a\psi}{\psi} &= \bigl( G(z)-1 \bigr) \frac{\dh_a\tau}{\tau} \,,\\
\frac{\dd_a\bar\psi}{\bar\psi} &=	\bigl( \bar{G}(z)-1 \bigr) \frac{\dh_a\tau}{\tau} \,.
\end{split}
\end{equation}
If\/ $\dd_a$ and\/ $\dd_b$ correspond to\/ $\dh_a$ and\/ $\dh_b$, respectively, then\/
$[\dd_a,\dd_b]$ corresponds to\/ $[\dh_a,\dh_b]$ up to adding a function from\/ $\F$.
\end{proposition}
\begin{proof}
The proof is almost identical to the proof of \cite[Theorem 5.1]{vM91}.
First, observe that by \eqref{dda2}, we have
\begin{equation}\label{psiw}
\frac{\dd_a\psi}{\psi} = \frac{\dd_a w}{w} \,, \qquad \frac{\dd_a\bar\psi}{\bar\psi} = \frac{\dd_a \bar w}{\bar w} \,,
\end{equation}
where the actions on $w$ and $\bar w$ are induced by the actions on $W$ and $\W$
(see \eqref{wavop}, \eqref{wavfn}).
Then \eqref{vftau} follows from \eqref{EBTHtaudress} by calculating the logarithmic derivatives and using that $\dd_a$
commutes with $G(z)$ and $\bar{G}(z)$, because $\dd_a$ commutes with the flows $\ps=\p{x_0}$, $\p{t_i}$ and $\p{\bar t_i}$.

After computing the commutator, we obtain:
\begin{align*}
\frac{[\dd_a,\dd_b]\psi}{\psi} &= \bigl( G(z)-1 \bigr) \frac{[\dh_a,\dh_b]\tau}{\tau} \,,\\
\frac{[\dd_a,\dd_b]\bar\psi}{\bar\psi} &=	\bigl( \bar{G}(z)-1 \bigr) \frac{[\dh_a,\dh_b]\tau}{\tau} \,.
\end{align*}
Then we note that $\Ker( G(z)-1 ) \cap \Ker( \bar{G}(z)-1 ) = \F$.
\end{proof}

Let us stress again that the vector fields $\dh_a$ on $\T$ are defined only up to adding a function $f\in\F$, because
\begin{equation*}
\bigl( G(z)-1 \bigr) \frac{(\dh_a+f)\tau}{\tau} = \bigl( G(z)-1 \bigr) \frac{\dh_a\tau}{\tau} \,,\qquad\quad f\in\F,
\end{equation*}
and there is a similar equation for $\bar{G}(z)$. Moreover, since $\tau$ and $f\tau$ are identified as tau-functions,
$\dh_a$ must have the property
\begin{equation*}
\bigl( G(z)-1 \bigr) \frac{\dh_a(f\tau)}{f\tau} = \bigl( G(z)-1 \bigr) \frac{\dh_a\tau}{\tau} \,,\qquad\quad f\in\F,
\end{equation*}
and the same with $\bar{G}(z)$. Therefore,
\begin{equation}\label{vftau2}
[\dh_a,f] \in\F \quad\;\text{for all}\quad a\in\V_0+\V_1, \; f\in\F.
\end{equation}

\subsection{Heisenberg symmetries of the EBTH}\label{sheis}

Now we will determine the vector fields $\dh_a$ on $\T$ that correspond to $\dd_a$ ($a\in\V_0$) via the Adler--Shiota--van Moerbeke formula \eqref{vftau}. Recall that $\V_0$ is linearly spanned by the following functions from $\E$ to $\A$:
\begin{equation*}
L^{\frac{n}{k}}, \;\; L^{\frac{n}{m}}, \;\; 2L^p\log L \qquad (n\in\Z,\; p\in\Z_{\ge0}).
\end{equation*}

\begin{lemma}\label{lheis1}
For every\/ $n\ge1$ and\/ $p\ge0$, we have
\begin{equation*}
\dh_{L^{\frac{n}{k}}} = \p{t_n}, \quad \dh_{L^{\frac{n}{m}}} = \p{\bar t_n}, \quad  
\dh_{2L^p\log L} = \p{x_p}, \quad \dh_{L^0} = x_0 =s, 
\end{equation*}
up to adding a function from\/ $\F$.
\end{lemma}
\begin{proof}
Using 
 \eqref{EBTHdresflows}, \eqref{EBTHtaudress}, \eqref{psiw},
we find
\begin{equation*}
\frac{\dd_{L^{\frac{n}{k}}}\psi}{\psi} = \frac{\dd_{L^{\frac{n}{k}}} w}{w} 
= \frac{\p{t_n} w}{w}  = \p{t_n} \log\Bigl( \frac{G(z)\tau}{\tau} \Bigr)
= \bigl( G(z)-1 \bigr) \frac{\p{t_n}\tau}{\tau} \,,
\end{equation*}
and similarly for the action on $\bar\psi$. This means that $\dh_{L^{\frac{n}{k}}} = \p{t_n}$.
The next two equations are proved in the same way.
Finally,
\begin{equation*}
\frac{\dd_{L^0}\psi}{\psi} = 0 = \bigl( G(z)-1 \bigr) \frac{s\tau}{\tau} \,, \qquad
\frac{\dd_{L^0}\bar\psi}{\bar\psi} = 1 = \bigl( \bar{G}(z)-1 \bigr) \frac{s\tau}{\tau} \,,
\end{equation*}
since $(L^0)_-=0$ and $(L^0)_+=1$.
\end{proof}

The vector fields $\p{t_n}$, $\p{\bar t_n}$ and $\p{x_p}$ from the above lemma are just infinitesimal translations of the corresponding variables when acting on $\tau\in\T$.
On the other hand, if we replace $\tau$ with $\exp(c\dh_{L^0})\tau = e^{cs}\tau$
for $c\in\CC$, the wave function $w$ remains unchanged while $\bar w$ becomes
$e^c\bar w$. This amounts to rescaling $\W\mapsto e^c\W$, which gives the same Lax operator 
$L$ according to \eqref{BTHLaxdress}.

\begin{lemma}\label{lheis2}
For every\/ $n\ge1$, we have
\begin{equation*}
\dh_{L^{\frac{-n}{k}}} = n t_n, \qquad
\dh_{L^{\frac{-n}{m}}} = n\bar t_n,
\end{equation*}
up to adding a function from\/ $\F$.
\end{lemma}
\begin{proof}
By definition, $L^{\frac{-n}{k}} = W\Lambda^{-n} W\inv\in\A_-$. Hence,
\begin{equation*}
\dd_{L^{\frac{-n}{k}}}\psi = -(L^{\frac{-n}{k}})_-\psi = -W\Lambda^{-n} W\inv\psi 
= -z^{-n} \psi
\end{equation*}
and
\begin{equation*}
\dd_{L^{\frac{-n}{k}}}\bar\psi = (L^{\frac{-n}{k}})_+\bar\psi = 0.
\end{equation*}
This matches with
\begin{equation*}
\bigl( G(z)-1 \bigr) n t_n = -z^{-n}, \qquad
\bigl( \bar{G}(z)-1 \bigr) n t_n = 0.
\end{equation*}
The proof for $L^{\frac{-n}{m}}$ is similar.
\end{proof}

The symmetries from \leref{lheis2} give rise to the transformations
\begin{align*}
\exp\Bigl( \sum_{n=1}^\infty \Bigl( c_n\dd_{L^{\frac{-n}{k}}} 
+ \bar c_n\dd_{L^{\frac{-n}{m}}}\Bigr)\Bigr) \psi 
&= \exp\Bigl( -\sum_{n=1}^\infty c_nz^{-n} \Bigr) \psi, \\
\exp\Bigl( \sum_{n=1}^\infty \Bigl( c_n\dd_{L^{\frac{-n}{k}}} 
+ \bar c_n\dd_{L^{\frac{-n}{m}}}\Bigr)\Bigr) \bar\psi 
&= \exp\Bigl( \sum_{n=1}^\infty \bar c_nz^{n} \Bigr) \bar\psi, 
\end{align*}
which correspond to
$W\mapsto WC$ and $\W\mapsto \W\bar C$
where
\begin{equation*}
C = \exp\Bigl( -\sum_{n=1}^\infty c_n\La^{-n} \Bigr), \quad
\bar C = \exp\Bigl( \sum_{n=1}^\infty \bar c_n\La^{n} \Bigr), \qquad
c_n,\bar c_n\in\CC.
\end{equation*}
These transformations of the dressing operators $W$, $\W$ preserve the
Lax operator $L$ from \eqref{BTHLaxdress}.

The above discussion shows that the symmetries $\dh_a$ $(a\in\V_0)$ 
are essentially trivial. Finally, observe that the equality 
$L^{\frac{nk}{k}}=L^{\frac{nm}{m}}$ for $n\in\Z_{\ge1}$ does not imply that
$(\p{t_{nk}}-\p{\bar t_{nm}})\tau = 0$ for $\tau\in\T$. It only implies that
\begin{equation}\label{heistau}
(\p{t_{nk}}-\p{\bar t_{nm}})\tau = f_n\tau \quad\;\text{for some}\quad f_n\in\F.
\end{equation}
In particular, when we replace $\tau$ with
\begin{equation*}
\exp\Bigl( \sum_{n=1}^\infty \Bigl( c_n\dh_{L^{\frac{-n}{k}}} 
+ \bar c_n\dh_{L^{\frac{-n}{m}}}\Bigr)\Bigr) \tau 
= \exp\Bigl( \sum_{n=1}^\infty \Bigl( n c_n t_n
+ n \bar c_n \bar t_n \Bigr)\Bigr) \tau,
\end{equation*}
this will add $nk c_{nk} - nm \bar c_{nm}$ to each $f_n$.

\subsection{Virasoro symmetries of the EBTH}\label{svir}

In this subsection, we determine the vector fields $\dh_a$ on $\T$ that correspond to $\dd_a$ ($a\in\V_1$) via the
Adler--Shiota--van Moerbeke formula \eqref{vftau}. In other words, we want to find $\dh_{L^p(M-\M)}$ for $p\ge0$,
where $M-\M$ is given by \eqref{MMbar2}.
First, we see how they commute with the 
Heisenberg symmetries from the previous subsection.

\begin{lemma}\label{lvir1}
For every\/ $n\ge1$ and\/ $p,q\ge0$, we have$:$
\begin{align*}
\bigl[ \p{t_n} , \dh_{L^p(M-\M)} \bigr] 
&= \begin{cases}
\frac{n}k \p{t_{(p-1)k+n}},  &p\ge1 \;\;\text{or}\;\; p=0, \; n>k, \\
\frac{n}k (k-n) t_{k-n},  &p=0, \; n<k, \\
x_0,  &p=0, \;n=k,
\end{cases} \\[8pt]
\bigl[ \dh_{L^p(M-\M)}, t_n \bigr] 
&= \begin{cases}
\frac{1}k \p{t_{(p-1)k-n}},  &n<(p-1)k, \\
\frac{1}k (n-(p-1)k) t_{n-(p-1)k},  &n>(p-1)k, \\
\frac{1}k x_0,  &n=(p-1)k,
\end{cases} \\[8pt]
\bigl[ \p{\bar t_n} , \dh_{L^p(M-\M)} \bigr] 
&= \begin{cases}
\frac{n}m \p{\bar t_{(p-1)m+n}},  &p\ge1 \;\;\text{or}\;\; p=0, \; n>m, \\
\frac{n}m (m-n) \bar t_{m-n},  &p=0, \; n<m, \\
x_0,  &p=0, \;n=m,
\end{cases} \\[8pt]
\bigl[ \dh_{L^p(M-\M)}, \bar t_n \bigr] 
&= \begin{cases}
\frac{1}m \p{\bar t_{(p-1)m-n}},  &n<(p-1)m, \\
\frac{1}m (n-(p-1)m) \bar t_{n-(p-1)m},  &n>(p-1)m, \\
\frac{1}m x_0,  &n=(p-1)m,
\end{cases} \\[8pt]
\bigl[ \p{x_q} , \dh_{L^p(M-\M)} \bigr] 
&= \begin{cases}
q \p{x_{p+q-1}} + \frac{1}k \p{t_{(p+q-1)k}} + \frac{1}m \p{\bar t_{(p+q-1)m}},  &p+q>1, \\
q \p{x_0} + \bigl( \frac1k+\frac1m\bigr) x_0,  &p+q=1, \\
t_k + \bar t_m,  &p=q=0,
\end{cases} \\[8pt]
\bigl[ \dh_{L^p(M-\M)}, x_q \bigr] &= 0,
\end{align*}
up to adding a function from\/ $\F$.
\end{lemma}
\begin{proof}
The equation $[\dh_{L^p(M-\M)},x_q]\in\F$ for $q\ge1$ is a consequence of \eqref{vftau2}.
The remaining ones follow immediately from the commutation relations \eqref{comm1}--\eqref{comm3}, \prref{pvftau}, 
and Lemmas \ref{lheis1}, \ref{lheis2}.
\end{proof}

We introduce linear operators $D_p$ ($p\ge0$) on $\T$ that satisfy the same
commutation relations with the Heisenberg symmetries as $\dh_{L^p(M-\M)}$ in the above lemma. Let
\begin{align*}
D_0 &= 
\frac1k \sum_{n=1}^\infty \bigl( x_{n+1} \p{t_{nk}} + (n+k) t_{n+k} \p{t_n} \bigr)
+ \frac1{2k} \sum_{n=1}^{k-1} n t_n (k-n) t_{k-n} \\
&+ \frac1m \sum_{n=1}^\infty \bigl( x_{n+1} \p{\bar t_{nm}} 
+ (n+m) \bar t_{n+m} \p{\bar t_n} \bigr) 
+ \frac1{2m} \sum_{n=1}^{m-1} n \bar t_n (m-n) \bar t_{m-n} \\
&+ \sum_{n=1}^\infty n x_n \p{x_{n-1}}
+ \Bigl( \frac1k+\frac1m\Bigr) x_0 x_1
+ x_0 (t_k + \bar t_m), \\[8pt]
D_1 &= 
\frac1k \sum_{n=1}^\infty \bigl( x_n \p{t_{nk}} + n t_n \p{t_n} \bigr)
+ \frac1m \sum_{n=1}^\infty \bigl( x_n \p{\bar t_{nm}} 
+ n \bar t_n \p{\bar t_n} \bigr) \\
&+ \sum_{n=1}^\infty n x_n \p{x_n}
+ \Bigl( \frac1k+\frac1m\Bigr) \frac{x_0^2}2 \,, \\
\intertext{and for $p\ge1$,}
D_{p+1} &= 
\sum_{n=1}^\infty n x_n \p{x_{p+n}}
+ \frac1k \sum_{n=0}^\infty \bigl( x_n \p{t_{(p+n)k}} + n t_n \p{t_{pk+n}} \bigr)
+ \frac1{2k} \sum_{n=1}^{pk-1} \p{t_n} \p{t_{pk-n}} \\
&+ \frac1m \sum_{n=0}^\infty \bigl( x_n \p{\bar t_{(p+n)m}} 
+ n \bar t_n \p{\bar t_{pm+n}} \bigr) 
+ \frac1{2m} \sum_{n=1}^{pm-1} \p{\bar t_n} \p{\bar t_{pm-n}} .
\end{align*}

\begin{lemma}\label{lvir2}
All equations from \leref{lvir1} hold if we replace\/ $\dh_{L^p(M-\M)}$ with\/ $D_p$
for\/ $p\ge0$.
\end{lemma}
\begin{proof}
The proof is straightforward and left to the reader.
\end{proof}

Next, we will determine the vector field $\dh_{L^p(M-\M)}$ for $p=0$.

\begin{lemma}\label{lvir3}
The action of\/ $\dh_{M-\M}$ on\/ $\T$ is given by 
\begin{equation*}
\dh_{M-\M} = L_{-1} := D_0 + \frac{k-1}2 t_k - \frac{m+1}2 \bar t_m,
\end{equation*}
up to adding a function from\/ $\F$.
\end{lemma}
\begin{proof}
Recall that $M-\M$ is given by \eqref{MMbar2}, and that
$Ws\Lambda^{- k} W\inv \in\A_-$, $\W s\Lambda^m \W\inv\in\A_+$.
Using the expansion \eqref{wavop}, we find
\begin{equation*}
Ws	
= sW - \sum_{i=1}^\infty i w_i(s) \La^{- i}
\end{equation*}
and
\begin{align*}
(Ws\Lambda^{- k} W\inv) \psi &= Ws\Lambda^{- k} \chi = z^{- k} Ws \chi \\
&= z^{- k} s\psi - z^{- k} \sum_{i=1}^\infty i w_i(s) z^{- i} \chi \\
&= z^{- k} s\psi + z^{- k+1} (\p{z} w)\chi.
\end{align*}
Hence, by \eqref{EBTHtaudress}, \eqref{GGbar},
\begin{align*}
\frac{ -(Ws\Lambda^{- k} W\inv)_- \psi}\psi &= -sz^{- k} - z^{- k+1} \frac{\p{z} w}w \\
&= -sz^{- k} - \sum_{i=1}^\infty z^{-k-i} G(z) \Bigl( \frac{\p{t_i}\tau}\tau \Bigr).
\end{align*}

Now let us calculate the action of the other terms from \eqref{MMbar2} on $\psi$.
Using \eqref{vftau}, \eqref{heistau} and Lemmas \ref{lheis1}, \ref{lheis2}, we obtain for $n\ge1$:
\begin{align*}
\frac{ -x_{n+1} (L^n)_- \psi}{\psi} &= \bigl( G(z)-1 \bigr) \frac{x_{n+1} \p{t_{nk}} \tau}{\tau} 
= \bigl( G(z)-1 \bigr) \frac{x_{n+1} \p{\bar t_{nm}} \tau}{\tau} \,,\\
\frac{ -x_n (2 L^{n-1} \log L)_- \psi}{\psi} &= \bigl( G(z)-1 \bigr) \frac{x_n \p{x_{n-1}} \tau}{\tau} \,,\\
\frac{ -\bar{t}_n (L^{\frac{n-m}{m}})_- \psi}{\psi} &= \bigl( G(z)-1 \bigr) \frac{(m-n) \bar{t}_n \bar{t}_{m-n} \tau}{\tau} = 0,
\qquad 1\le n\le m-1, \\
\frac{ -\bar{t}_n (L^{\frac{n-m}{m}})_- \psi}{\psi} &= \bigl( G(z)-1 \bigr) \frac{\bar{t}_n \p{\bar{t}_{n-m}} \tau}{\tau} \,,
\qquad n\ge m+1.
\end{align*}
Next, note that
\begin{align*}
\frac{ -x_1 (L^0)_- \psi}{\psi} &= 0 = \bigl( G(z)-1 \bigr) \frac{x_1 x_0 \tau}{\tau} \,, \\
\frac{ -\bar{t}_m (L^0)_- \psi}{\psi} &= 0 = \bigl( G(z)-1 \bigr) \frac{\bar{t}_m x_0 \tau}{\tau} \,, \\
\frac{ -t_k (L^0)_- \psi}{\psi} &= 0 = \bigl( G(z)-1 \bigr) \frac{t_k x_0 \tau}{\tau} + \frac1k sz^{- k}.
\end{align*}
Finally, we have
\begin{align*}
-\frac1{\psi} \sum_{n=1}^{k-1} n t_n & (L^{\frac{n-k}{k}})_- \psi = -\sum_{n=1}^{k-1}  n t_n z^{n-k} \\
&= \frac12 \bigl( G(z)-1 \bigr) \Bigl( \sum_{n=1}^{k-1} n t_n (k-n) t_{k-n} + (k-1) k t_k \Bigr)
\end{align*}
and
\begin{align*}
-\frac1{\psi} \sum_{n=k+1}^\infty n t_n & (L^{\frac{n-k}{k}})_- \psi 
= \sum_{n=k+1}^\infty n t_n \bigl( G(z)-1 \bigr) \frac{\p{t_{n-k}} \tau}{\tau}  \\
&= \bigl( G(z)-1 \bigr) \sum_{n=k+1}^\infty \frac{n t_n \p{t_{n-k}} \tau}{\tau}
+ \sum_{n=k+1}^\infty z^{-n} G(z) \Bigl( \frac{\p{t_{n-k}} \tau}{\tau} \Bigr).
\end{align*}

Putting all of the above together, we obtain that
\begin{equation*}
-\frac{(M-\M)_- \psi}{\psi} = \bigl( G(z)-1 \bigr) \frac{L_{-1} \tau}{\tau} \,.
\end{equation*}
A similar calculation shows
\begin{equation*}
\frac{(M-\M)_+ \bar\psi}{\bar\psi} = \bigl( \bar G(z)-1 \bigr) \frac{L_{-1} \tau}{\tau} \,,
\end{equation*}
using that
\begin{align*}
\frac{ (\W s\Lambda^m \W\inv)_+ \bar\psi}{\bar\psi} 
= sz^{m} + z^{m+1} \frac{\p{z} \bar w}{\bar w} 
= sz^{m} + \sum_{i=1}^\infty z^{m+i} \bar{G}(z) \Bigl( \frac{\p{t_i}\tau}\tau \Bigr).
\end{align*}
This completes the proof of the lemma.
\end{proof}

Define the operators
\begin{align*}
L_{-1} &= D_0 + \frac{k-1}2 t_k - \frac{m+1}2 \bar t_m, \\
L_{0} &= D_1 - \Bigl( \frac1{2k}+\frac1{2m} \Bigr) x_0 + \frac{(k+m)(km+2)}{24km} \,, \\
L_{p} &= D_{p+1} - \frac1{2k} \p{t_{pk}} - \frac1{2m} \p{\bar t_{pm}}, \qquad p\ge1.
\end{align*}
The constant term in $L_0$ is added so that these commute as elements of the \emph{Virasoro algebra}.

\begin{lemma}\label{lvir4}
We have\/ $[L_p,L_q]=(p-q)L_{p+q}$ for all\/ $p,q\ge-1$.
\end{lemma}
\begin{proof}
This is a straightforward calculation.
\end{proof}

Now we can state our second main result.

\begin{theorem}\label{thm2}
Up to adding a function from\/ $\F$, we have\/
$\dh_{L^p(M-\M)} = L_{p-1}$ for\/ $p\ge0$.
\end{theorem}
\begin{proof}
The idea of the proof is to use the Lie brackets, and is similar to that of \cite[Theorem 3.2]{AvM95}.
As before, consider all operators modulo a function from $\F$.
Let $R_p=\dh_{L^p(M-\M)} - L_{p-1}$. Then we already know $R_0=0$, by \leref{lvir3}.
Due to Lemmas \ref{lvir1} and \ref{lvir2}, we have $[R_p,h]=0$ for all
\begin{equation*}
h\in H=\Span\{ x_{n-1},t_n,\bar t_n, \p{x_{n-1}}, \p{t_n}, \p{\bar t_n} \,|\, n\ge1 \},
\end{equation*}
which implies $R_p\in H$. 

On the other hand, from \eqref{comm4} and \leref{lvir4}, we get $[\dh_{M-\M},R_1]=0$
because $[L_0,L_{-1}]=L_{-1}$ and $[\dh_{L(M-\M)},\dh_{M-\M}]=\dh_{M-\M}$.
Using \leref{lvir1} again, we see that if $h\in H$ and $[\dh_{M-\M},h]=0$, then
$h\in\Span\{\p{t_{nk}}-\p{\bar t_{nm}} \,|\, n\ge1 \}$. By \eqref{heistau}, this means
that $h=0$ (mod $\F$). Therefore, $R_1=0$.

Now we prove by induction that $R_p=0$ for all $p\ge1$. Assuming $R_p=0$, we get
$[\dh_{M-\M},R_{p+1}]=0$
because $[L_p,L_{-1}]=(p+1)L_{p-1}$ and $[\dh_{L^{p+1}(M-\M)},\dh_{M-\M}]=(p+1)\dh_{L^p(M-\M)}$.
Then, as above, we conclude that $R_{p+1}=0$.
\end{proof}

\begin{remark}\label{rvir1}
Consider the operators 
\begin{equation*}
B_p = L^p(M-\M) - \frac{k-1}{2k} L^{\frac{(p-1)k}k} + \frac{m+1}{2m} L^{\frac{(p-1)m}m},
\qquad p\ge0.
\end{equation*}
Then $\dh_{B_p} = D_p + \de_{p,1} \frac{(k+m)(km-1)}{24km}$.
The advantage of this formulation is that one can make the substitution \eqref{ttbar}
in $D_p$;
see \cite[Chapter 4]{W15}.
\end{remark}

\begin{example}\label{evir2}
For $k=m=1$, after the substitution \eqref{ttbar}, the operators $D_p$ become:
\begin{align*}
D_0 &=	\sum_{n=1}^\infty \bigl( 2x_{n+1} \p{t_n} + (n+1) t_{n+1} \p{t_n} + nx_n\p{x_{n-1}} \bigr) 
+ x_0 t_1 + 2x_0 x_1, \\
D_1 &=	\sum_{n=1}^\infty \bigl( 2x_n\p{t_n} + nt_n\p{t_n} + nx_n\p{x_n} \bigr) + x_0^2, \\
D_{p+1} &= \sum_{n=0}^\infty \bigl( 2x_n\p{t_{n+p}} + nt_n\p{t_{n+p}} + nx_n\p{x_{n+p}} \bigr) 
+ \sum_{n=1}^{p-1} \p{t_n}\p{t_{p-n}}, \qquad p\ge1.
\end{align*}
Then making the change of variables from \reref{reth1}, we obtain:
\begin{align*}
L_{- 1}&=\sum_{n=1}^\infty \left( t^{1,n}\p{t^{1,n-1}} + t^{2,n}\p{t^{2,n-1}} \right) +
\frac{1}{\epsilon^2}t^{1,0}t^{2,0},\\
L_0&=\sum_{n=1}^\infty n \left(t^{1,n}\p{t^{1,n}} + t^{2,n-1}\p{t^{2,n-1}} \right) + 2t^{1,n}\p{t^{2,n-1}} + \frac{1}{\epsilon}\left(t^{1,0}\right)^2,\\[6pt]
L_p&=\epsilon^2\sum_{n=1}^{p-1}k!(p-n)! \, \p{t^{2,n-1}}\p{t^{2,p-n-1}} 
+ 2\sum_{n=0}^\infty \alpha_p(n) \, t^{1,n}\p{t^{2,p+n-1}}   \\
&\quad	+ \sum_{n=1}^\infty \frac{(p+n)!}{(n-1)!} \left( t^{1,n}\p{t^{1,p+n}} + t^{2,n-1}\p{t^{2,p+n-1}} \right),
\end{align*}
for $p\ge1$, where
\begin{equation*}
	\alpha_p(0)	=	p!, \qquad \alpha_p(n)	=	\frac{(p+n)!}{(n-1)!} \, \sum_{j=n}^{p+n}\frac{1}{j} \,, \quad n>0.
\end{equation*}
These operators coincide with the Virasoro symmetries of the ETH of 
Dubrovin--Zhang \cite[Eq.\ (1.29)]{DZ04}.
\end{example}

\section{Conclusions}\label{s5}

In this paper, we constructed additional symmetries of the extended bigraded Toda hierarchy (EBTH) and described explicitly their action on the Lax operator, wave operators, and tau-function of the hierarchy. In particular, we obtained infinitesimal symmetries of the EBTH that act on the tau-function as a subalgebra of the Virasoro algebra, generalizing those of Dubrovin--Zhang \cite{DZ04}.

A natural question is whether this representation extends to the whole Virasoro algebra.
We believe this is possible before making the reduction $t_{nk}=\bar t_{nm}$ $(n\ge1)$.
It appears that our Virasoro operators are closely related to those from \cite{B15}.
After the change of variables from \reref{rebth1}, they can also be related to 
the Virasoro operators from \cite{DZ99,EHX97,EJX98,Giv01}. 
We expect that the symmetries of the EBTH give Virasoro constraints for the total descendant potential of $\mathbb{CP}^1$ with two orbifold points (cf.\ \cite{DZ99,Giv01,MT08,CvdL13}). 
Another interesting question is whether these Virasoro constraints extend to more general $\mathcal W$-constraints as in \cite{AvM92,BM13}. We hope to investigate these and other related questions in the future.

\section*{Acknowledgements}
We are grateful to Guido Carlet and Todor Milanov for stimulating discussions
and to the referees for valuable suggestions and comments.
The first author was supported in part by a Simons Foundation grant.

\bibliographystyle{amsalpha}

\end{document}